\documentclass[letterpaper, 10 pt, conference]{ieeeconf}  

\usepackage{subcaption}
\let\labelindent\relax


\IEEEoverridecommandlockouts
\usepackage{microtype}
\usepackage{graphicx}
\usepackage{booktabs} 
\usepackage{wrapfig}
\usepackage[font=footnotesize]{caption}

\usepackage{enumitem}

\usepackage{xcolor}
\usepackage{algorithm}
\usepackage{algpseudocodex}

\usepackage{hyperref}
\usepackage{makecell}

\usepackage{amsmath}
\usepackage{amssymb}
\usepackage{mathtools}
\usepackage{amsthm}

\usepackage[capitalize,noabbrev]{cleveref}

\theoremstyle{plain}
\newtheorem{theorem}{Theorem}[section]
\newtheorem{proposition}[theorem]{Proposition}

\newtheorem{corollary}[theorem]{Corollary}
\theoremstyle{definition}

\newtheorem{assumption}[theorem]{Assumption}

\theoremstyle{remark}

\algrenewcommand\textproc[1]{\ensuremath{\mathrm{#1}}}

\title{\textbf{Inverse Safety Filtering: Inferring Constraints from Safety Filters for Decentralized Coordination}}

\author{Minh Nguyen$^{1}$, Jingqi Li$^{2}$, Gechen Qu$^{1}$, and Claire J. Tomlin$^{1}$
\thanks{*This work was supported by NSF Safe Learning Enabled Systems, 
the DARPA Assured Autonomy, ANSR, and TIAMAT programs, and
the NASA ULI on Safe Aviation Autonomy.  
J.L. was supported by an Oden
Institute Fellowship.}
\thanks{$^{1}$Minh Nguyen (corresponding author), Gechen Qu, and Claire J. Tomlin are with the Department of Electrical Engineering and Computer Sciences, University of California Berkeley, 
{\tt\small minh02@berkeley.edu, qugch@berkeley.edu, tomlin@berkeley.edu}.}%
\thanks{$^{2}$Jingqi Li is with the Oden Institute, University of Texas at Austin, 
        {\tt\small jingqi.li@austin.utexas.edu}.}%
}

\begin{document}
\maketitle

\begin{abstract}
Safe multi-agent coordination in uncertain environments can benefit from learning constraints from other agents. Implicitly communicating safety constraints through actions is a promising approach, allowing agents to coordinate and maintain safety without expensive communication channels. This paper introduces an online method to infer constraints from observing the safety-filtered actions of other agents. We approach the problem by using safety filters to ensure forward safety and exploit their structure to work backwards and infer constraints. We provide sufficient conditions under which we can infer these constraints and prove that our inference method converges. This constraint inference procedure is coupled with a decentralized planning method that ensures safety when the constraint activation distance is sufficiently large. We then empirically validate our method with Monte Carlo simulations and hardware experiments with quadruped robots.
\end{abstract}


%
\IEEEpeerreviewmaketitle

\section{Introduction}


Cooperative multi-agent systems are widely used in applications such as transportation \cite{yang_collaborative_2022}, area monitoring \cite{cortes_coverage_2004}, and manipulation \cite{khatib_coordination_1996}, where teams of robots coordinate to accomplish tasks that are impractical for a single agent. In these settings, constraints such as obstacle avoidance are essential for the safe operation of the entire team. Centralized planning can be utilized to coordinate teams while respecting constraints, but is often impractical due to unreliable communication, limited bandwidth, and computational bottlenecks. Decentralized planning mitigates these issues by enabling agents to operate using local information, but introduces challenges with asymmetric observations of the environment. This raises a fundamental question: \emph{how can decentralized agents coordinate safely under information asymmetry?}

\begin{figure}[t]
    \centering
    \includegraphics[width=\linewidth]{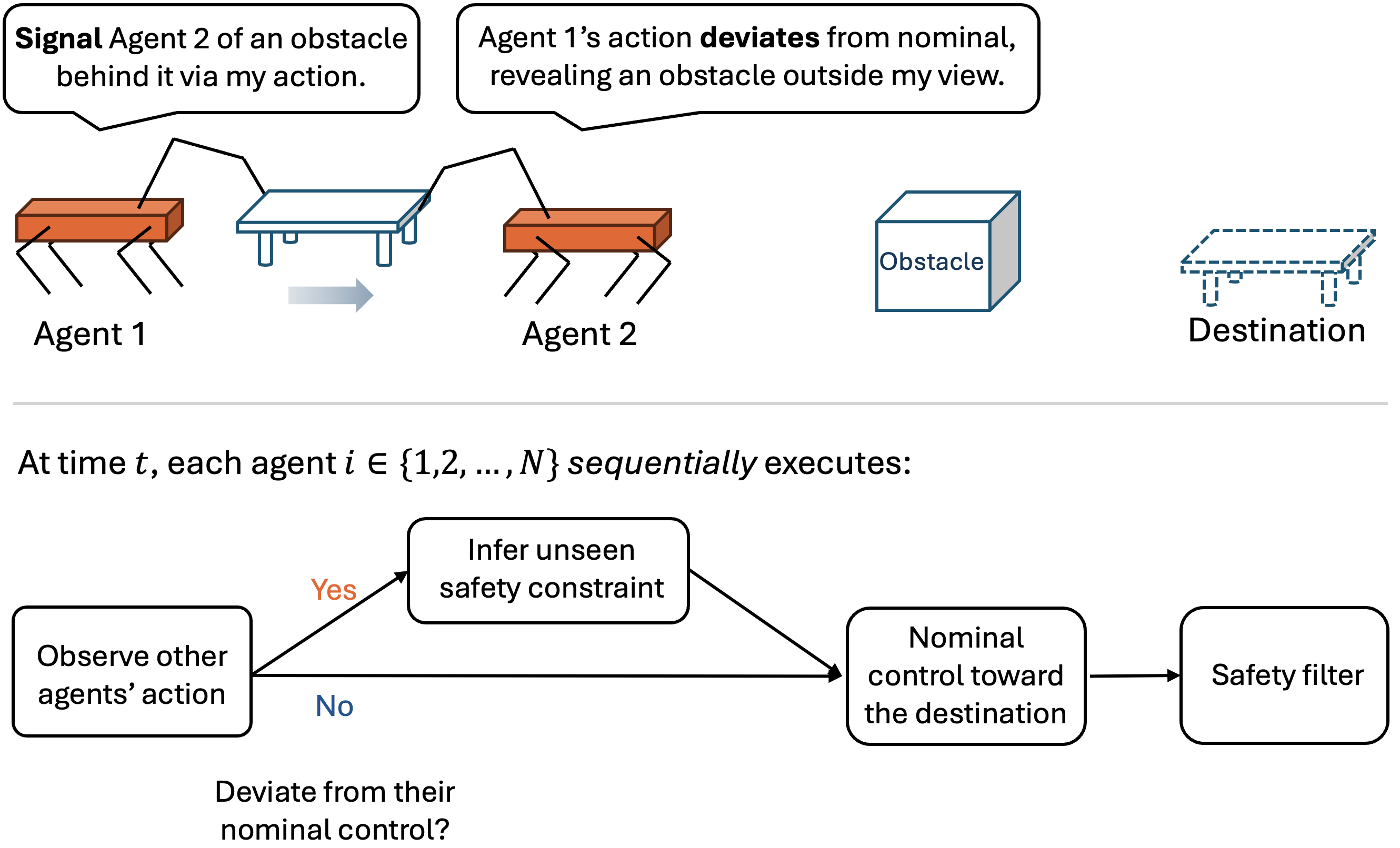}
    \caption{Applying our decentralized constraint inference and planning method to a furniture-moving task. Agents share a goal but possess asymmetric knowledge of environmental constraints (Agent 1 initially knows of the obstacle but Agent 2 does not). Our method enables decentralized constraint inference and safe planning, with theoretical guarantees on constraint identifiability and safety without direct communication among agents.}
    \label{fig:highleveldiagram}
\end{figure}
A substantial body of work addresses safe multi-agent coordination by equipping each agent with a Control Barrier Function (CBF) safety filter \cite{ames_control_2014}. The goal is to minimally modify a nominal control input to ensure forward invariance of the safe set. Most prior work \cite{borrmann_control_2015, wang_distributed_2010} assumes agents already know their teammates' constraints. Separately, the problem of inferring objectives and constraints from demonstrations has been studied through inverse optimal control \cite{papadimitriou_constraint_2023} and inverse reinforcement learning \cite{ng2000algorithms, malik_inverse_2021}. These methods require offline expert demonstrations. The intersection of CBF-based safety filtering and online constraint inference in multi-agent settings remains largely unexplored.

In this work, we show that the optimality conditions of CBF-based safety filters can be exploited to efficiently and accurately infer other agents’ constraints. Our key insight is that \emph{enforcing a contraction condition on the constraint values over time, as introduced in standard control barrier function conditions, improves the identifiability of the constraint and therefore facilitates constraint inference.} We further prove that our proposed inference procedure converges to the true constraint given that the constraint activation distance is sufficiently large, providing formal guarantees for the constraint learning procedure. By combining constraint learning with a decentralized planning method, we enable safe and coordinated operation of agent teams without explicit communication channels.

Our contributions in this work are:
\begin{enumerate}
    \item A provably convergent method for inferring constraints from safety filters. 
    \item A generalization of the constraint inference method in 1) to safety filters involving multi-agent formation constraints.
    \item A decentralized inference and planning method for multiple agents to safely cooperate without direct communication, handling static and moving obstacles, as well as multi-team interactions. Moreover, we establish safety guarantees for this decentralized control framework. 
    \item Hardware experiments on quadruped robots that demonstrate the real-world applicability of the approach.
\end{enumerate}

\section{Related Work}
\subsection{Inferring missing information in decentralized control}

Inferring information about other agents has been extensively researched in inverse optimal control \cite{molloy_online_2020} and inverse reinforcement learning \cite{ng2000algorithms, self_model-based_2022}. These techniques have been used offline \cite{li_cost_2023} and online \cite{soltanian_peer-aware_2025} to learn about the intent, or objectives, of other agents. Other work has inferred the constraints of other agents by optimizing over the optimality conditions of observed actions \cite{palafox_learning_2023, papadimitriou_constraint_2023}, sampling possible parameters of the constraints \cite{doi:10.1177/02783649211035177}, and constrained reinforcement learning \cite{scobee_maximum_2019, malik_inverse_2021, xu2024uncertaintyaware}. However, these approaches are primarily offline and require expert demonstration data. Our approach provides an online, model-based method to infer constraints from multi-agent interactions. We study conditions under which our constraint inference method is guaranteed to converge and empirically characterize the set of initial conditions that converge to the true solution.

\subsection{Safe planning methods}

Many methods have been used to ensure safe control of dynamical systems such as barrier certificates \cite{prajna_safety_2004} and Hamilton-Jacobi reachability \cite{bansal_hamilton-jacobi_2017}. A popular approach is to use CBFs \cite{ames_control_2014}, which are attractive because they provide a reasonable control law to enforce constraints. CBF constraints can be used in minimally invasive safety filters \cite{ames_control_2019, hsu_safety_2024, wabersich_data-driven_2023} and integrated into trajectory planners \cite{zeng_safety-critical_2021}, which showcases their versatility. These approaches have been used for safe, cooperative planning in quadruped \cite{kim_layered_2023} and UAV applications \cite{pallar_optimal_2025, hegde_multi-uav_2022}. In this work, we exploit the structure of CBF safety filters as both a forward safety filter and a mechanism from which to efficiently infer constraints from.

\subsection{Multi-agent planning under incomplete information}

Recent approaches to decentralized, multi-agent planning have used consensus optimization \cite{alonso-mora_distributed_2019} and CBF filters \cite{wang_safety_2016} to generate collision-free trajectories for teams of agents. A key challenge in this area is learning from other agents to fill in missing information about the environment \cite{dragan_legibility_2013, doi:10.1177/02783649211035177}.  Our approach is similar to Bujarbaruah et al. \cite{bujarbaruah_learning_2021}, in which we implement a demonstrator-learner dichotomy but with a more general formulation of constraints and no fixed control policy. This allows us to handle larger teams and ensure formal safety guarantees.

\section{Problem Formulation}
\label{sec:problem_formulation}

Consider an $N$-player, $T$-stage, deterministic, discrete-time dynamic game, with a state $x_t^i \in \mathbb{R}^{n_i}$ and control input $u_t^i \in \mathbb{R}^{m_i}$ for each player $i \in [N] := \{1, \cdots, N\}$, $t \in [T]$. Let the dimension of the joint state and control input be $n := \sum_{i=1}^N n_i$ and $m := \sum_{i=1}^N m_i$, respectively. We denote by $x_t := [x_t^i]_{i=1}^N \in \mathbb{R}^n$ and $u_t := [u_t^i]_{i=1}^N \in \mathbb{R}^m$ the joint state and joint control at time $t \in [T]$, respectively. Define $\mathbf{x}:=[x_t]_{t=1}^T$ and $\mathbf{u}:=[u_t]_{t=1}^T$. The joint dynamics is:
\begin{equation}\label{eq:dynamics}
    x_{t+1} = f(x_t) + g(x_t) \cdot u_t
\end{equation}
The objective of each agent $i$ is to minimize its overall cost, given by the sum of its running costs $c_t^i : \mathbb{R}^n \times \mathbb{R}^m \to \mathbb{R}$ over the time horizon $T$:
\begin{equation}
    J^i(\mathbf{x}, \mathbf{u}) := \sum_{t=1}^{T} c_t^i(x_t, u_t)
\end{equation}

\subsection{Parameterized Constraints}

To simplify our analysis of constraint identifiability, we define the \textit{constraint-relevant state} at each timestep as $s^i_t \coloneq P x^i_t \in \mathbb{R}^d$ via a projection $P : \mathbb{R}^{n_i} \to \mathbb{R}^{d}$, and the \textit{constraint-relevant control matrix} as $B_s \coloneq P \cdot g(x_t^i)$. We consider constraints of the form $h(s, \theta) \geq 0$ parameterized by $\theta \in \mathbb{R}^k$, such that inferring $\theta$ is sufficient for specifying the entire constraint. For example, $s$ could be the position of an agent and $\theta$ could be the position of an obstacle.

To enforce these constraints, we utilize Discrete-Time CBFs which take the form of:
\begin{equation}
\label{eq:cbf_condition}
    h(s_{t+1}, \theta) \geq (1 - \gamma)\, h(s_t, \theta),
\end{equation}
where $\gamma \in (0, 1)$ is the decay rate. This condition ensures the barrier value at the next timestep is lower-bounded by $(1-\gamma)$ times its current value. This constraint is enforced through a \emph{safety filter} that minimally modifies a given nominal control:
\begin{equation}
\label{eq:safety_filter}
\begin{aligned}
    u_{\mathrm{safe}} = \arg\min_{u} \quad & \tfrac{1}{2}\|u - u_{\mathrm{nom}}\|^2 \\
    \text{s.t.} \quad & h(s_{t+1}(u), \theta) \geq (1-\gamma)\, h(s_t, \theta)
\end{aligned}
\end{equation}
For a control-affine system with a quadratic constraint as in \eqref{eq:barrier}, the problem is a non-convex Quadratically Constrained Quadratic Program (QCQP).

\textbf{Running Example:} We consider a furniture-moving problem similar to Figure~\ref{fig:highleveldiagram} where two agents work together to carry a piece of furniture past obstacles to a goal position. The robots are modeled as double integrators:
\begin{equation}
    p_{t+1} = p_t + \Delta t \cdot v_t + \tfrac{1}{2}\Delta t^2 \cdot u_t, \quad
    v_{t+1} = v_t + \Delta t \cdot u_t.
\end{equation}
The constraint-relevant state, $s$, is the position of each agent, and obstacle avoidance is enforced via circular constraints parameterized by the centerpoint of the obstacle, $\theta$:
\begin{equation}
\label{eq:barrier}
    h(s, \theta) = (s - \theta)^\top Q (s - \theta) - r^2,
\end{equation}
where $Q \succ 0$ is a symmetric, positive definite matrix and $r > 0$ is a safety radius. For circular obstacles, $Q = I$ gives $h(s, \theta) = \|s - \theta\|^2 - r^2$. This is a common constraint enforcing that an agent must stay outside a region centered on~$\theta$.

\subsection{Decentralized Constraint Inference Problem}

\textbf{Constraint Inference Problem Formulation.} Given an observation of the nominal action, the filtered safe action, and the current state $(u_{\mathrm{nom}}, u_{\mathrm{safe}}, x_t)$, find the constraint parameterized by $\theta$ that explains the control deviation $\Delta u = u_{\mathrm{safe}} - u_{\mathrm{nom}}$ of another agent.

In Sections~\ref{sec:decoding} and~\ref{sec:multi_agent}, we present our method for inferring constraints from safety filters. When only the CBF constraint is active, we derive a closed-form solution for the constraint parameter. When additional constraints such as formation keeping are simultaneously active, we design a provably convergent Newton solver for numerical inference.

\section{Inverting CBF Safety Filters}
\label{sec:decoding}

We first show how one can infer a constraint parameterized by $\theta$ from a safety filter with an active CBF constraint and then provide sufficient conditions under which said constraint is identifiable and unique.

\subsection{KKT Optimality Conditions}
\label{sec:kkt}

To solve the constraint inference problem, we exploit the structure of the safety filter~\eqref{eq:safety_filter}. We first derive the Karush–Kuhn–Tucker (KKT) conditions of the filter and then obtain a closed-form solution for its unknown parameter~$\theta$.

The Lagrangian of the safety filter~\eqref{eq:safety_filter} with dual variable $\lambda \geq 0$ is:
\begin{equation}
    \label{eq:stationarity}
    \begin{split}
        \mathcal{L}(u, \lambda) = &\tfrac{1}{2}\|u - u_{\mathrm{nom}}\|^2 \\
                                &- \lambda\bigl[h(s_{t+1}(u), \theta) - (1{-}\gamma)h(s_t, \theta)\bigr].
    \end{split}
\end{equation}
At an optimal solution, $u_{\mathrm{safe}}$, the KKT stationarity condition gives:
\begin{equation}
\label{eq:kkt_stationarity}
    \Delta u \coloneqq u_{\mathrm{safe}} - u_{\mathrm{nom}} = \lambda\, \nabla_u h(s_{t+1}, \theta).
\end{equation}
For the remainder of this paper, we focus on quadratic constraints, which are widely used for obstacle avoidance and can approximate complex geometries in robotic applications. However, our framework extends to more general constraints.

Let $h(s, \theta) = (s - \theta)^\top Q (s - \theta) - r^2$ be a quadratic barrier and let $s_{t}$ evolve according to control-affine dynamics. Then $\nabla_u h = 2 B_s^\top Q(s_{t+1} - \theta)$ by the chain rule.

Substituting the gradient into~\eqref{eq:kkt_stationarity} yields:
\begin{equation}
\label{eq:kkt_combined}
    \Delta u = 2\lambda\, B_s^\top Q(s_{t+1} - \theta).
\end{equation}
The obstacle lies along the direction:
\begin{equation}
\label{eq:direction}
    \hat{d} \coloneqq
    \frac{Q^{-1}(B_s^\top)^{-1}\Delta u}
         {\|Q^{-1}(B_s^\top)^{-1}\Delta u\|}.
\end{equation}

\subsection{Closed-Form Solution}
\label{sec:closed_form}

With the relationship between the control deviation and constraint, we can derive a closed-form expression for $\theta$.
\begin{theorem}[Closed-Form Solution]
\label{thm:closed_form}
Suppose the safety filter constraint is active, $B_s \in \mathbb{R}^{d
\times d}$ is invertible, and $Q \succ 0$. Define $e \coloneqq s_t - s_{t+1}$. Define $A = \gamma\, \hat{d}^\top Q \hat{d}$, $B = -2(1{-}\gamma)(e^\top Q \hat{d})$, and $C = -(1{-}\gamma)\|e\|_Q^2 - \gamma r^2$.
The obstacle position satisfies:
\begin{equation}
\label{eq:obstacle_position}
    \theta = s_{t+1} - t^* \hat{d}, \quad t^* = \frac{-B + \sqrt{B^2 - 4AC}}{2A}.
\end{equation}
\end{theorem}
\begin{proof}
From the stationarity condition~\eqref{eq:kkt_stationarity} and the
quadratic barrier gradient, $\Delta u = 2\lambda B_s^\top Q(s_{t+1} -
\theta)$. Since $B_s$ is invertible and $Q \succ 0$, left-multiplying
by $Q^{-1}(B_s^\top)^{-1}$ and normalizing gives the obstacle
direction $\hat{d}$ in~\eqref{eq:direction}. Parameterizing
$\theta = s_{t+1} - t\hat{d}$ and substituting into the active
constraint $h(s_{t+1}, \theta) = (1-\gamma)\,h(s_t, \theta)$ yields
$At^2 + Bt + C = 0$ with the stated coefficients. Since $A > 0$ and
$C < 0$, the quadratic has exactly one positive root, which is the
unique obstacle position.
\end{proof}

For linear systems (e.g., double integrator), $B_s$ is constant and known analytically. For nonlinear systems, $B_s = P \cdot g(x_t)$ can be recomputed at every planning point. The derivation is for quadratic CBF constraints, but we provide sufficient conditions for general identifiability and uniqueness in the next section.

\subsection{Identifiability and Uniqueness}
\label{sec:uniqueness}

Identifiability means the observation $(u_{\mathrm{nom}}, u_{\mathrm{safe}}, x_t)$ carries enough information to solve for $\theta$. In general, multiple constraints can explain the same control deviation, which can cause coordination issues within teams. We now give sufficient conditions for identifiability and show that the quadratic barrier structure further guarantees uniqueness.

\begin{theorem}[Constraint Identifiability]
\label{thm:identifiability}
Consider a safety filter with control $u \in \mathbb{R}^m$, constraint-relevant state $s \in \mathbb{R}^d$, and a twice-differentiable barrier $h(s, \theta)$ with parameter $\theta \in \mathbb{R}^k$. The parameter $\theta$ is locally identifiable from a single observation if:
\begin{enumerate}[nosep]
    \item The constraint is active ($\lambda > 0$),
    \item $\operatorname{rank}(B_s) \geq k$ (sufficient actuation),
    \item $\operatorname{rank}(\nabla^2_{s\theta} h) \geq k$ (barrier sensitivity to $\theta$).
\end{enumerate}
\end{theorem}
\begin{proof}
The KKT stationarity~\eqref{eq:kkt_stationarity} and binding constraint form $m + 1$ equations in $k + 1$ unknowns $(\theta, \lambda)$. Differentiating the stationarity condition with respect to $\theta$ gives Jacobian $\lambda B_s^\top \nabla^2_{s\theta} h$, whose rank is at least $k$ by Conditions~2--3 (with $\lambda > 0$ from Condition~1). By the implicit function theorem, $\theta$ is locally determined as a function of $\lambda$, and the binding constraint provides the remaining scalar equation to determine~$\lambda$.
\end{proof}

\begin{corollary}[Global Uniqueness for Quadratic Barriers]
\label{cor:uniqueness}
Under Conditions~1--2 of Theorem~\ref{thm:identifiability}, if we have a quadratic barrier constraint \eqref{eq:barrier}, then Condition~3 holds automatically ($\nabla^2_{s\theta} h = -2Q$), and the recovered $\theta$ is globally unique.
\end{corollary}
\begin{proof}
The linearity of $\nabla_s h = 2Q(s - \theta)$ in $\theta$ reduces the stationarity condition to a ray $\theta = s_{t+1} - t\hat{d}$ for $t > 0$. Substituting into the binding constraint yields the quadratic in Theorem~\ref{thm:closed_form}, whose sign structure $A > 0$, $C < 0$ guarantees exactly one positive root by Descartes' rule.
\end{proof}


\section{Inverting Safety Filters under Formation Constraints}
\label{sec:multi_agent}
The previous section considers constraint recovery from the perspective of an external observer. In a full $N$-agent system, however, agents must balance safety and inference, which complicates the problem since control deviations may arise from multiple active constraints (e.g., obstacle avoidance and formation keeping). In this section, we introduce a more realistic safety filter for multi-agent systems and show that the desired constraint can still be inferred via a provably convergent method.

\subsection{Formation Constraints}
\label{sec:formation}

To handle multi-agent formations, we enforce pairwise distance constraints between agents, which is common in applications such as swarming or cooperative manipulation:
\begin{equation}
\label{eq:formation}
    (d - \epsilon)^2 \leq \|s^i - s^j\|^2 \leq (d + \epsilon)^2,
\end{equation}
where $d$ is the target formation distance and $\epsilon > 0$ is the formation slack. Each pairwise distance is enforced via two inequality constraints: a lower bound $g_{\mathrm{lower}} = \|s^i - s^j\|^2 - (d-\epsilon)^2 \geq 0$ and an upper bound $g_{\mathrm{upper}} = (d+\epsilon)^2 - \|s^i - s^j\|^2 \geq 0$. Since only one formation constraint can be active at a time, we denote the active formation constraint as $g_{\mathrm{form}}(s_{t+1}) \geq 0$ and add it to the safety filter formulation~\eqref{eq:safety_filter}.


\subsection{Inference with Multiple Active Constraints}
\label{sec:multiple_constraints}

By adding additional constraints, the safety filter can no longer be inverted using the method of the previous section. Instead, we propose a provably convergent method to infer constraints from the new filter.

When both obstacle and formation constraints are active, the KKT stationary condition becomes:
\begin{equation}
\label{eq:kkt_multi}
    \Delta u = \lambda\, \nabla_u h(s_{t+1}, \theta) + \nu\, \nabla_u g_{\mathrm{form}},
\end{equation}
where $\lambda, \nu \geq 0$ are the obstacle and formation multipliers, respectively. The formation gradient for agent $i$ with respect to agent $j$ is:
\begin{equation}
    \nabla_u g_{\mathrm{form}} = 2 B_s^\top (s_{t+1}^i - s_{t+1}^j) \eqqcolon 2 B_s^\top f.
\end{equation}
Note that the sign of the gradient will depend on whether the lower or upper constraint is active.

Denoting $c = s_{t+1} - \theta$, this yields a nonlinear system $F(x) = 0$ with $x = [\theta^\top, \lambda, \nu]^\top$:
\begin{equation}
\label{eq:newton_system}
F(x) = \begin{bmatrix}
\Delta u - 2\lambda B_s^\top Q\, c - 2\nu B_s^\top f \\
c^\top Q\, c {-} (1{-}\gamma)(s_t {-} \theta)^\top Q (s_t {-} \theta) {-} \gamma r^2
\end{bmatrix} = 0
\end{equation}

Row~1 provides $m$ stationarity equations; Row~2 provides the active constraint equation. With $m = d$ (fully actuated), the system has $m + 1$ independent equations for $d + 2$ unknowns ($\theta$, $\lambda$, $\nu$), leaving one degree of freedom.

\subsection{A Provably Convergent Newton Solver for Constraint Inference}
\label{sec:newton_convergence}

We solve the underdetermined system of equations \eqref{eq:newton_system} via the regularized least-squares problem:
\begin{equation}
\label{eq:regularized_newton}
    \min_x \;\mathcal{L}(x) \coloneqq \tfrac{1}{2}\|F(x)\|^2 + \tfrac{\mu}{2}\|x - x_0\|^2,
\end{equation}
where $x_0$ is the initialization from the closed-form obstacle-only solution (Theorem~\ref{thm:closed_form}) and $\mu > 0$ is a regularization parameter. The first order optimality condition of \eqref{eq:regularized_newton} is
\begin{equation}
\label{eq:newton_optimality}
    \pi(x) \coloneqq \nabla_x F(x)^\top F(x) + \mu(x - x_0) = 0.
\end{equation}
We utilize an inexact Newton method to find the solution to $\pi(x)=0$. At iteration $k$, we form a first-order Taylor expansion of $\pi$ around the current iterate $x_k$, yielding the linear system:
\begin{equation}
    \nabla \pi(x_k)\,\Delta x_k = -\pi(x_k).
\end{equation}
The resulting direction $\Delta x_k$ is used to update the iterate via
\begin{equation}
    x_{k+1} = x_k + \alpha_k \Delta x_k,
\end{equation}
where the step size $\alpha_k \in (0,1]$ is chosen using a standard backtracking line search to ensure sufficient decrease in the residual norm $\|\pi(x)\|$. We now provide conditions in which Newton's method converges to the true constraint parameter.

\begin{theorem}[Newton Convergence]
\label{thm:newton_convergence}
Let $\mathcal{S} = \{x : \mathcal{L}(x) \leq \mathcal{L}(x_0)\}$ be the sublevel set of the initial iterate.
Suppose $\mu > M F_{\max}$ where $M = \max_i \|\nabla^2 F_i\|$ and
$F_{\max} = \max_{x \in \mathcal{S}} \|F(x)\|_1$. Then Newton iterates
on~\eqref{eq:newton_optimality} with backtracking line search satisfy
$\|\pi(x_k)\| \to 0$.
\end{theorem}
\begin{proof}
The objective $\mathcal{L}$ is coercive since $\mu > 0$, so the
sublevel set $\mathcal{S}$ is compact (and closed). On $\mathcal{S}$, $\nabla\pi = J^\top J + \mu I + \sum_i F_i \nabla^2 F_i \succeq (\mu - M
F_{\max})I \succ 0$, giving a bounded inverse. Since $\nabla\pi$ is polynomial, it is Lipschitz on $\mathcal{S}$. The result follows from Section 10.3.3 in Boyd and Vandenberghe \cite{boyd_convex_2023}.
\end{proof}

In practice, additional constraints (e.g., formation keeping) can produce control deviations even without an active obstacle constraint. To avoid false inferences, we reject deviations below a noise threshold $\|\Delta u\| < \epsilon_{\Delta u}$, check whether the deviation is fully explained by formation constraints alone, and verify the inferred parameter is not a known obstacle. After these checks, we select the closed-form solution or the Newton solver based on the number of active constraints.

\begin{algorithm}[t]
\caption{Decentralized Planning with Inference}
\label{alg:planning}
\begin{algorithmic}[1]
\Require Joint initial state $x_1$, obstacle beliefs $\Theta^{1:N}$, planning horizon $T$, and task execution horizon $T_e$
\While{$t\le T_e$}
    \For{each demonstrating agent $i \in [N]$}
        \For{each learning agent $j$ at $t-1$}\label{alg:start get nominal control for inference} 
        \State Let $k$ be the demonstrating agent at $t-1$
            \If{$u^k_{\mathrm{safe},t-1} \neq u^k_{\mathrm{nom},t-1}$} 
                \State $\hat{\theta} \!\gets\! \mathrm{Inference}(u_{\mathrm{nom},t-1}^k, u_{\mathrm{safe,t-1}}^k, x_{t-1})$ 
                \State $\Theta^j \gets \Theta^j \cup \{ \hat{\theta} \}$ \label{alg:end get nominal control for inference}
            \EndIf
        \EndFor
        \State $u_{\mathrm{nom},t} \gets \mathrm{RecedingHorizonPlanner}(x_{t}, T)$\label{alg:start_plan_action}
        \State $u_{\mathrm{safe,t}}^i \gets \mathrm{SafetyFilter}(x_t, u_{\mathrm{nom}, t}, \Theta^i)$ 
        \For{each learning agent $j$ at $t$} 
            \State $u_{\mathrm{safe},t}^j \gets \mathrm{SafetyFilter}(x_t, u_{\mathrm{nom}, t}, \cap_{k=1}^N \Theta^k)$\label{alg:end_plan_action}
        \EndFor
        \State $x_{t+1} \gets f(x_t) + g(x_t) \cdot u_{\mathrm{safe},t}$ 
        \State $t \gets t+1$
    \EndFor
\EndWhile
\Statex
\Procedure{RecedingHorizonPlanner}{$x_1, T$}
    \State Solve the $N$-agent game problem over horizon $T$:
    \State \quad $\mathbf{x}^*,\mathbf{u}^* \gets \arg\min_{\mathbf{u},\mathbf{x}}  J^i(\mathbf{x}, \mathbf{u})$ \quad s.t. dynamics \eqref{eq:dynamics}
    \State \Return $u_1^*$ \\\Comment{Output the first step control as the nominal control}
\EndProcedure
\end{algorithmic}
\end{algorithm}

\section{Decentralized Round-Robin-Style Inference And Planning}

In this section, we introduce a method to run inference in decentralized multi-agent systems while guaranteeing safety, both against static and moving obstacles.

To solve the safety filter with formation constraints, agent $i$ must know the \textit{next} position of other agents. However, when agents simultaneously avoid private obstacles, determining the next position of other agents is not always possible. To address this, we use a round-robin execution framework where the team is split into one \textbf{demonstrating} agent and $N-1$ \textbf{learning} agents at each time step.
\begin{assumption}[Role Awareness]
We assume that all agents know who the demonstrating and learning agents are at every time step.
\end{assumption}
\subsection{Actions of Demonstrating and Non-Demonstrating Agents}

Each agent $i$ has knowledge of a set of constraints $\Theta^i$. A constraint is \emph{private} if it is known to a strict subset of agents and \emph{public} if it is known to all agents after demonstration. We assume all agents know the set of public constraints: $\cap_i^{N} \Theta^i$.

At each time step of the round-robin framework, to make the decentralized control and inference problem tractable, only the demonstrating agent acts on its private information, while learning agents act solely on public information (e.g., avoiding the nearest demonstrated obstacle while maintaining formation constraints). We summarize this framework in Algorithm~\ref{alg:planning}: lines \ref{alg:start get nominal control for inference}–\ref{alg:end get nominal control for inference} perform constraint inference using data from the previous time step, while lines \ref{alg:start_plan_action}–\ref{alg:end_plan_action} plan safe actions for the current time step.
The $\mathrm{Inference}$ procedure utilizes Theorem \ref{thm:closed_form} or Theorem \ref{thm:newton_convergence}, depending on the number of active constraints. The $\mathrm{SafetyFilter}$ procedure solves the safety filter problem \eqref{eq:safety_filter} for the closest obstacle in the corresponding agent’s constraint set, together with the formation constraints~\eqref{eq:formation}.

We establish the safety guarantee of our method as follows.


\begin{theorem}[Decentralized Safety Guarantee]
\label{thm:decentralized_safety}
Consider $N$ agents executing Algorithm~\ref{alg:planning} with $t > N$,
circular CBF constraints ($Q = I$, decay rate
$\gamma$) and formation constraints~\eqref{eq:formation} with target
distance $d$ and slack $\varepsilon$. Suppose: (1) demonstrating agents enforce the CBF constraint with \emph{inflated} radius $r_{\mathrm{demo}} = r + (d + \varepsilon)$. (2) Every agent's initial state satisfies the formation constraints, and lies in the forward invariant set~\cite{wabersich_data-driven_2023} of $h(s_0^i, \theta) \geq 0$, for each ground truth constraint's parameter $\theta$, where $h$ uses $r_{\mathrm{demo}}$. Then $\|s_t^i - \theta\| \geq r$ for all agents $i$, all $t$, and all ground truth constraints' parameters.
\end{theorem}

\begin{proof}
The condition $t > N$ ensures each agent has demonstrated once, enabling others to learn its constraints. By Condition~2 and CBF forward invariance, at each time $t$, the demonstrating agent's state $s_t^\mathrm{demo}$ satisfies $\|s_t^{\mathrm{demo}} - \theta\| \geq r_{\mathrm{demo}} = r + (d + \varepsilon)$ for all $t$.
For each learning agent $j$, the formation constraint (Condition~1) to the demonstrating agent and triangle inequality give:
\begin{equation}
    \begin{split}
        \|s_t^{j} - \theta\| &\geq \|s_t^{\mathrm{demo}} - \theta\| - \|s_t^{j} - s_t^{\mathrm{demo}}\| \\
        &\geq (r + d + \varepsilon) - (d + \varepsilon) = r.
    \end{split}
\end{equation}
Thus, learners remain at distance $\geq r$ from all unsafe sets at every timestep, even before inference. When the demonstrating agent's CBF activates (producing $\Delta u \neq 0$), the learning agents infer $\theta$ via Theorem~\ref{thm:closed_form} or~\ref{thm:newton_convergence} and add it to their constraint sets. At inference, the learning agent $j$ satisfies $\|s_t^{j} - \theta\| \geq r$, so the learners start in the safe set of their new constraint and forward invariance holds thereafter.
\end{proof}

\subsection{Generalizing to Moving Obstacles and Multi-Team Interactions}

In addition to avoiding static obstacles, we consider the case where multiple teams of agents are in the same environment. Agents within a team can treat other teams as moving obstacles and solve a robust safety filter to ensure safety. We assume that the demonstrating agent knows the position and velocity of the obstacle.

We assume there are $K$ teams. When an agent in team $k\in[K]$ observes team $l \in [K]$, team $l$'s center of agents is treated as a moving obstacle with bounded velocity $\|v_{\mathrm{obs}}\| \leq v_{\max}$. Let $p_t^k$ denote team $k$'s position and $p_t^l$ the observed position of the other team's agent. Note that in this case, the constraint-relevant state is the position of an agent, meaning the constraint is a function of $p$ and $\theta$. The next-step position of team $l$ is unknown from agent team $k$'s perspective:
\begin{equation}
    p_{t+1}^l = p_t^l + \Delta t \cdot v_t^l, \quad \|v_t^l\| \leq v_{\max}.
\end{equation}

Safety requires that the CBF condition holds for the \emph{worst-case} obstacle motion:
\begin{equation}
\label{eq:robust_cbf}
    \min_{\|v_t^l\| \leq v_{\max}} h(p_{t+1}^k,\, p_{t+1}^l(v_t^l)) \;\geq\; (1-\gamma)\,h(p_t^k,\, p_t^l).
\end{equation}

\begin{proposition}[Robust CBF via Worst-Case Velocity]
\label{thm:robust_cbf}
Let $h(p, \theta) = \|p - \theta\|^2 - r^2$ be a circular barrier ($Q = I$) and let $c = p_{t+1}^k - p_t^l$ be the predicted displacement between teams. The worst-case obstacle velocity that minimizes the left-hand side of~\eqref{eq:robust_cbf} is
\begin{equation}
\label{eq:worst_case_v}
    v_t^{l*} = \begin{cases}
        v_{\max} \cdot \dfrac{c}{\|c\|} & \text{if } \|c\| > \Delta t \cdot v_{\max}, \\[6pt]
        \dfrac{c}{\Delta t} & \text{otherwise}.
    \end{cases}
\end{equation}
In the first case, the resulting worst-case barrier value is $h_{\min} = (\|c\| - \Delta t\, v_{\max})^2 - r^2$. In the second case, $h_{\min} = -r^2 < 0$ (the obstacle can reach the agent).
\end{proposition}

\begin{proof}
The worst-case velocity minimizes $\|c - \Delta t\, v\|^2$ over $\|v\| \leq v_{\max}$. This is the projection of $c / \Delta t$ onto the ball of radius $v_{\max}$: if $\|c\|/\Delta t \leq v_{\max}$, the unconstrained minimizer $v = c/\Delta t$ is feasible and achieves zero residual; otherwise, $v^* = v_{\max} \cdot c / \|c\|$ moves directly toward the agent at maximum speed, giving residual $\|c\| - \Delta t\, v_{\max}$.
\end{proof}

A simpler alternative to solving the robust condition is to replace the robust CBF with a standard CBF using an \emph{inflated safety radius} that conservatively accounts for worst-case obstacle motion following the approach in \cite{borrmann_control_2015}:
\begin{equation}
\label{eq:inflated_radius}
    r_{\mathrm{robust}} = r_{\mathrm{safe}} + \Delta t \cdot v_{\max}.
\end{equation}
The demonstrating agent then solves the standard safety filter~\eqref{eq:safety_filter} with $r_{\mathrm{robust}}$ in place of $r$. In practice, we can utilize $r_\mathrm{robust}$ in the safety filter of Algorithm \ref{alg:planning} and inherit the robust CBF properties of avoiding moving obstacles.

\section{Experiments and Results}

We test our approach in both simulation and hardware to quantify inference accuracy and planning safety. Each experiment uses iLQGames \cite{fridovich2020efficient} as a nominal trajectory planner and circular CBF constraints for obstacle avoidance. The demonstrating agent filters its action according to the nearest obstacle in its private and public information while the learning agents filter their actions based solely on public information. We evaluate the following hypotheses:

\textit{\textbf{H1}. Our inference method and constraint formulation lead to a higher success rate than standard methods.} 

These ``standard'' methods include using a non-CBF constraint to enforce obstacle avoidance and optimizing over the forward filter as discussed in Equation \eqref{eq:input_matching}.

\textit{\textbf{H2}. Our method generalizes and scales to multi-agent interaction and multiple teams with decentralized inference and planning.}

\textit{\textbf{H3}. Our method provides a larger region of convergence and lower inference error than baseline methods.}

\subsection{Hypothesis 1: Our inference method and constraint formulation lead to a higher success rate than standard methods.}

We utilize a Monte-Carlo style rollout of $100$ trajectories and evaluate them based on the number of collisions, false inferences (ghosts), inference error to the true constraint, and the discovery rate of true constraints.

Each scenario has a two-agent team navigating through a 2D environment with $2$ or $3$ randomly placed obstacles. Each agent is modeled as a 2D, double integrator with a formation constraint and an obstacle avoidance constraint. We compare our inference method, labeled KKT in the figure, against an Input Matching (IM) approach where a solver tries to match the observed $u_\mathrm{safe}$ action by optimizing directly over the $\mathrm{SafetyFilter}$ constraint parameter argument:
\begin{equation}
    \label{eq:input_matching}
    \begin{split}
        \min_\theta || u_\mathrm{safe}  - \textrm{SafetyFilter}(x, u_\mathrm{nom}, \theta)||_2^2
    \end{split}
\end{equation}
To solve the IM optimization problem, we utilize a Nelder-Mead optimizer with 100 iterations to keep it within real-time planning constraints. For obstacle avoidance, we compare our CBF constraint against a non-CBF circular constraint in Equation \eqref{eq:barrier}, labeled Circle in the figure.

\begin{figure}[t]
    \centering
    \includegraphics[width=1.0\linewidth]{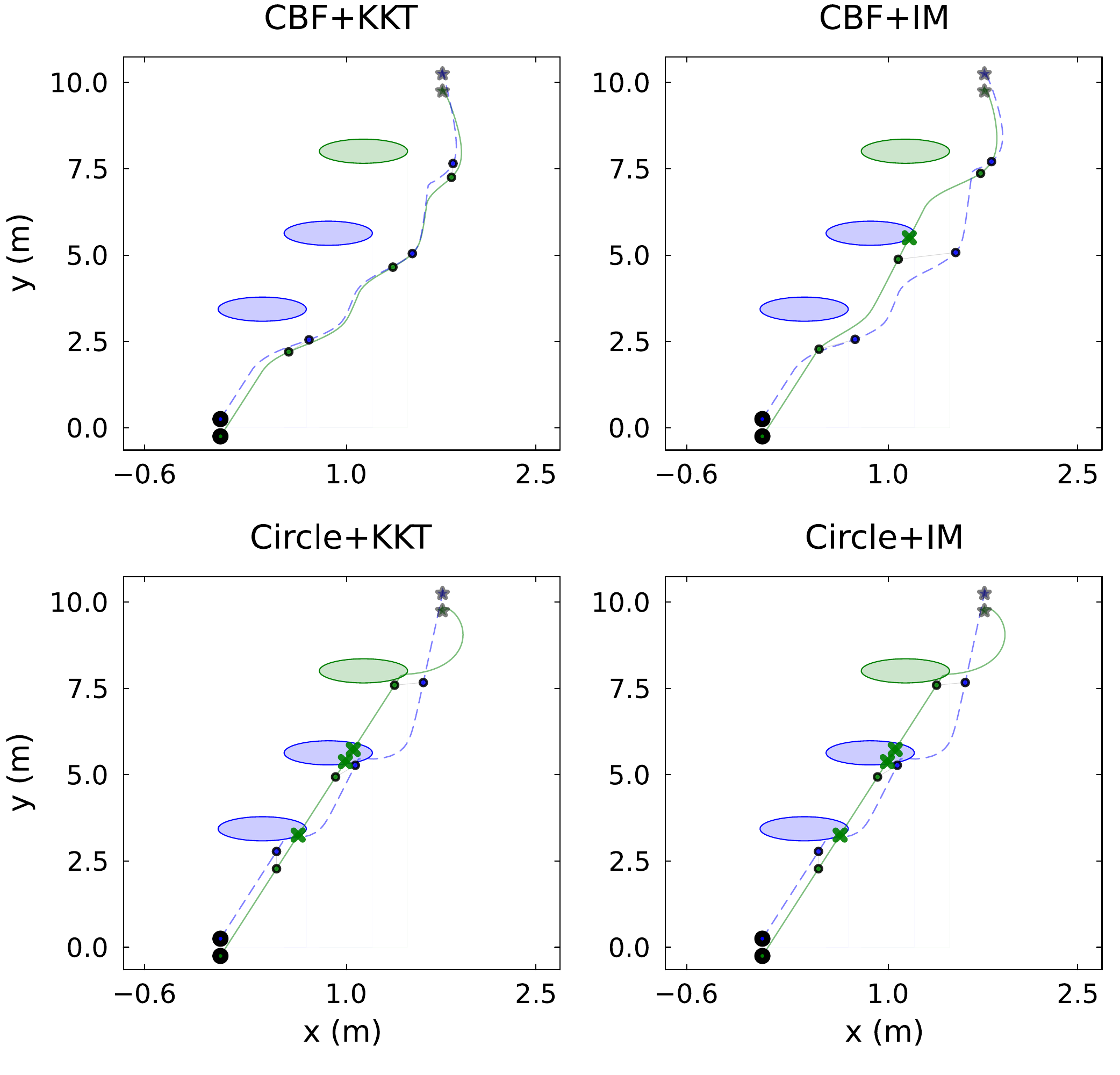}\vspace{-0.8em}
    \caption{A single trial from a Monte-Carlo rollout where we generate random obstacles and goal positions for two double integrator agents connected by a formation constraint. We compare two inference methods (KKT vs. IM) and two obstacle-avoidance formulations (CBF vs. Circle). We observe that CBF + KKT is able to avoid all collisions while the other methods have at least one collision.}
    \label{fig:mc_twoagent}
\end{figure}

{\setlength{\tabcolsep}{2.5pt}
\begin{table}[t]
\centering
\caption{Monte Carlo 100-Rollout Results}
\label{tab:results}
\begin{tabular}{lccccc}
\toprule
Configuration & Collisions $\downarrow$ & Ghosts $\downarrow$ & Error $\downarrow$ & Discovery $\uparrow$ \\
\midrule
\makecell[l]{CBF+KKT (Ours)} & $\mathbf{0.1 \pm 0.6}$ & $\mathbf{0.0 \pm 0.0}$ & $\mathbf{0.001 \pm 0.008}$ & $\mathbf{90\%}$ \\
CBF+IM & $0.4 \pm 1.8$ & $7.4 \pm 8.5$ & $1.497 \pm 1.235$ & $80\%$ \\
Circle+KKT & $5.6 \pm 3.7$ & $0.0 \pm 0.0$ & $0.002 \pm 0.011$ & $19\%$ \\
Circle+IM & $6.7 \pm 3.7$ & $0.0 \pm 0.0$ & --- & $0\%$ \\
\bottomrule
\end{tabular}
\end{table}
}

The trials show that with the same CBF obstacle avoidance constraints, the KKT method and IM method were both able to avoid collisions. Our method achieves significantly lower position error and a higher discovery rate than the IM approach. When using the Circle constraint formulation, both inference methods struggle. This is mainly due to the demonstrating agent only avoiding its obstacle after the learning agent has already collided, as seen at the bottom left of Figure \ref{fig:mc_twoagent}. The KKT inference method still outperforms IM even with these constraints, inferring $19\%$ of obstacles versus $0\%$ for IM across all trials.

\subsection{Hypothesis 2: Our method generalizes and scales to multiple agent teams and multiple teams in one environment.}

In addition to two-agent teams, we consider teams of three and four agents and demonstrate the ability of our round-robin framework to maintain formation constraints between larger teams of agents while still avoiding collisions.

In Figure \ref{fig:large_team_experiment}, each agent is modeled as a 2D, double integrator and has pairwise formation constraints to every other agent in the environment. The figure shows how each agent is able to demonstrate and learn from each other's constraints while avoiding collisions.

Next, we consider the case when we have multiple teams in the same environment. We utilize the radius described in \eqref{eq:inflated_radius} to ensure we do not collide with other, moving teams. To make sure agents stay below $v_\mathrm{max}$, we introduce another constraint in the safety filter on the velocity of the agents. We do not run inference when this constraint is active.

\begin{figure}
    \centering
    \includegraphics[width=1.0\linewidth]{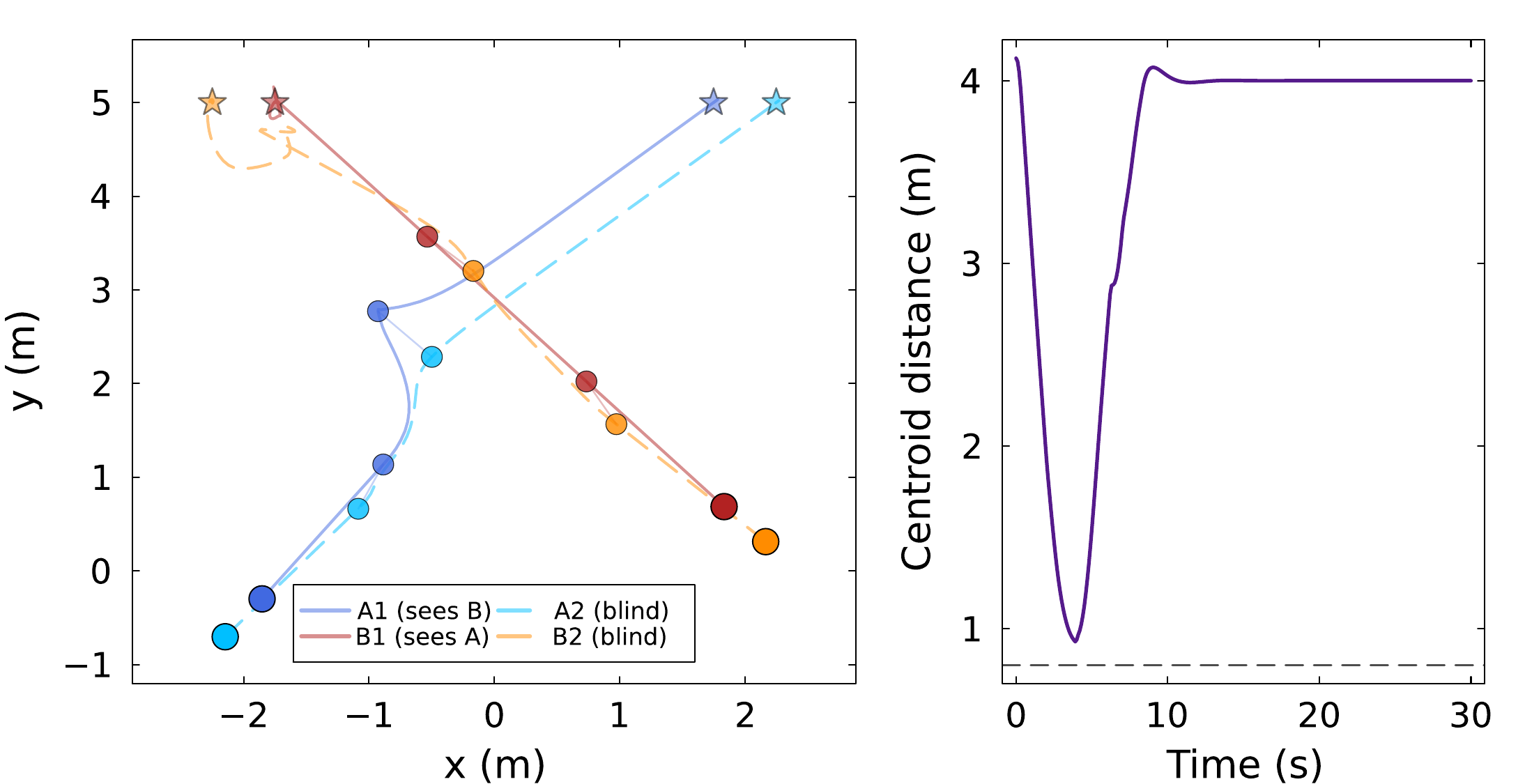}\vspace{-0.4em}
    \caption{Two teams of two agents cross diagonally toward opposite goals (left). The robust CBF constraint enables decentralized collision-free crossing, with the left team yielding. The distance plot (right) confirms the inter-team distance never drops below the safety threshold.}
    \label{fig:multi_team}
\end{figure}

\begin{figure}
\centering
\includegraphics[width=1.0\linewidth]{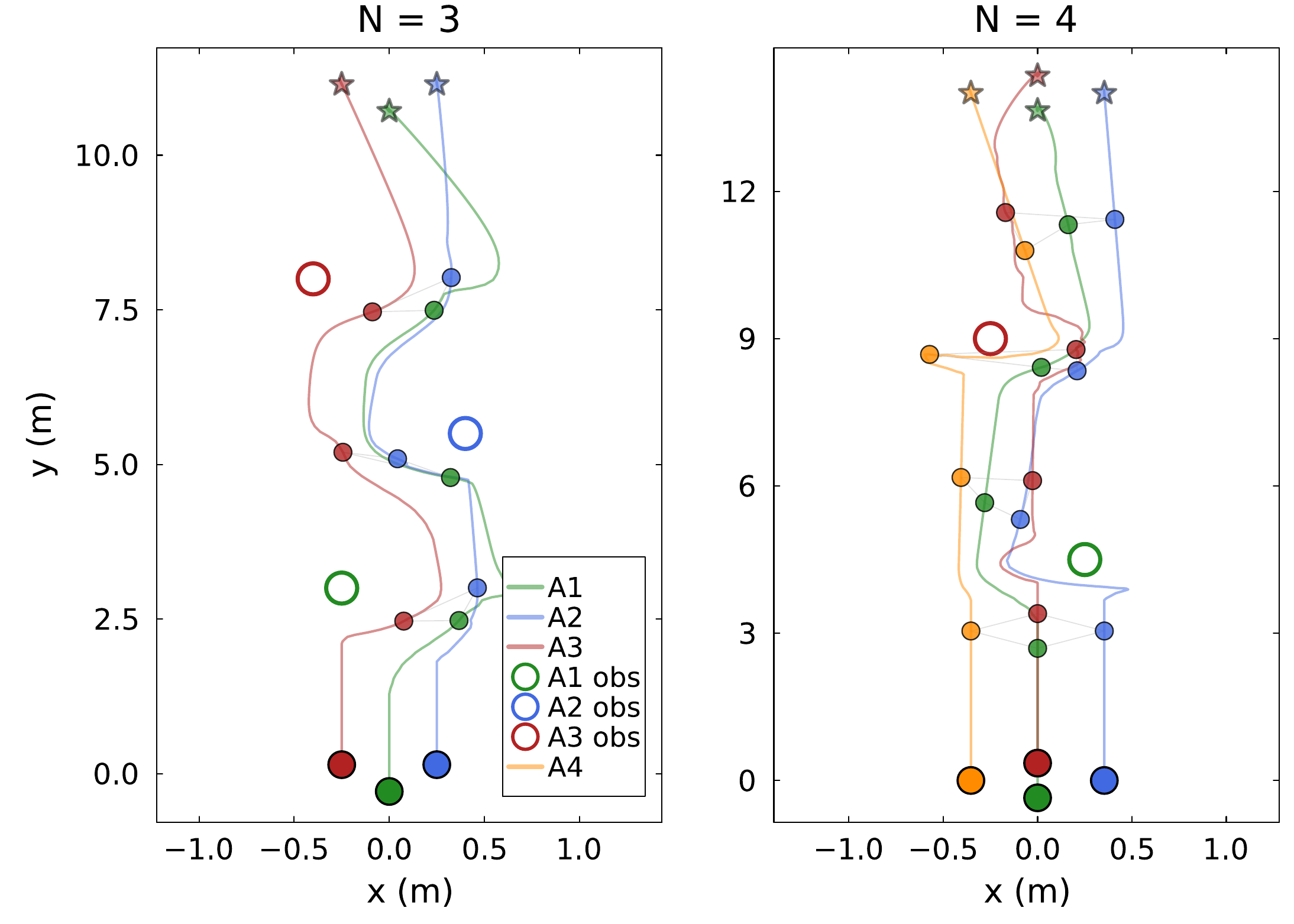}\vspace{-0.4em}
\caption{Two simulations showing 3 (left) and 4 (right) agent teams navigating through obstacles. The color of the obstacles matches with the agent that knows of it. Both scenarios are able to infer and avoid the obstacles, successfully reaching their goal positions.}
\label{fig:large_team_experiment}
\end{figure}

In Figure \ref{fig:multi_team}, both teams are able to cross towards their diagonal goals and avoid colliding with each other. The red and orange teams make it to the midpoint first, so the blue and purple teams treat them as an obstacle and yield before reaching their goal. This demonstrates the ability of our method to avoid moving obstacles and the soundness of our robust CBF formulation.

\subsection{Hypothesis 3: Our method provides a larger region of convergence and lower inference error than baseline methods.}

When obstacle and formation constraints are both active, we have an underdetermined system that we propose solving using Newton's method. We test this scenario by placing two agents such that their formation and obstacle avoidance constraints are both active. We then initialize Newton's method, Equation \eqref{eq:regularized_newton}, and the IM method with initial guesses around the learning agent.

Figure \ref{fig:newton_convergence} demonstrates that the region of convergence for Newton's method is quite large. This means that with a suitable initialization scheme, we can correctly infer the obstacle position. The IM method only converges to the true position inside a small circle around the true obstacle and converges to wrong positions or diverges outside.

\begin{figure}[t]
    \centering
    \includegraphics[width=1.0\linewidth]{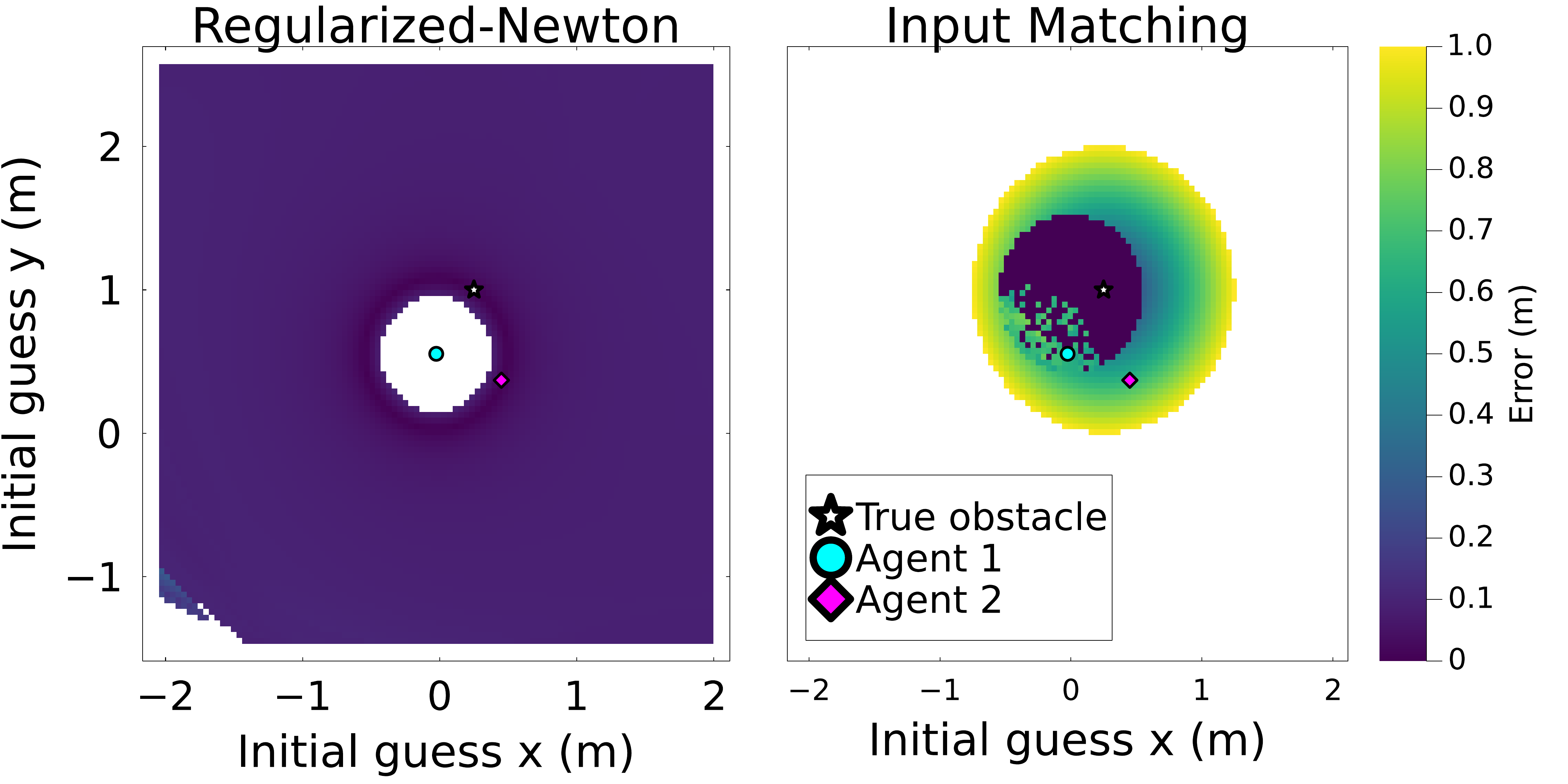}\vspace{-0.4em}
    \caption{
    Convergence regions of the regularized Newton method (left) and Input Matching (right). White denotes divergence. Newton's method converges from most initializations, whereas Input Matching requires initialization near the true obstacle.
    }
    \label{fig:newton_convergence}
\end{figure}
\begin{figure}[t]
    \centering
    \includegraphics[width=\linewidth]{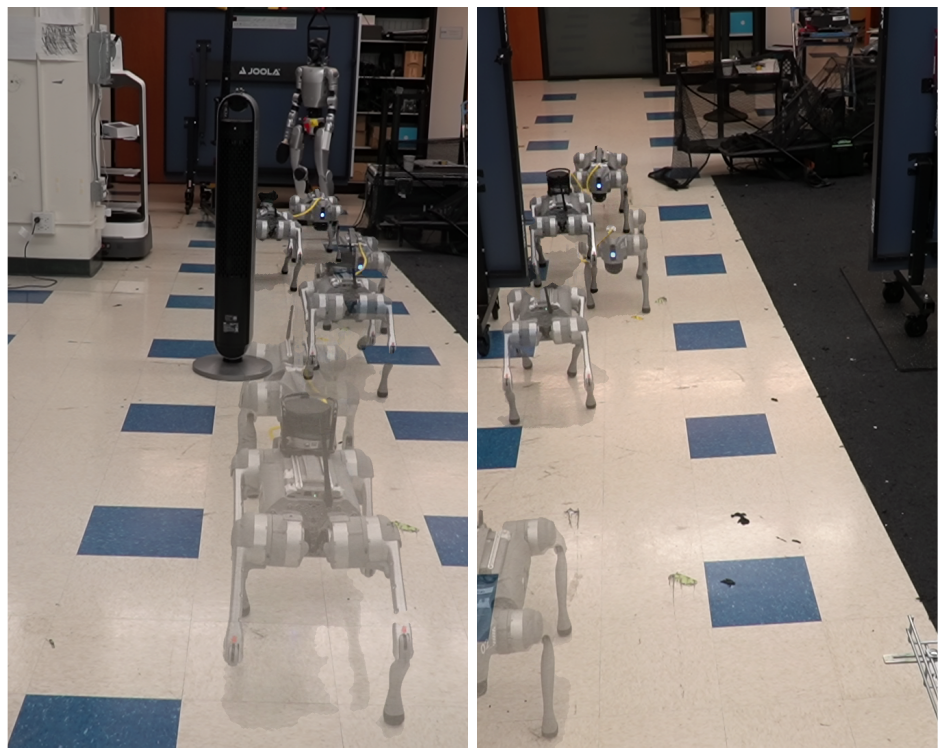}\vspace{-0.4em}
    \caption{
    Two rope-connected quadrupeds navigate around an obstacle (left) and through a gap between two walls (right). The rear robot knows the obstacles, while the front robot infers them from the rear robot’s safety-filtered actions.
    }
    \label{fig:hardware_keyframes}
\end{figure}
\begin{figure}[t]
    \centering
    \includegraphics[width=\linewidth]{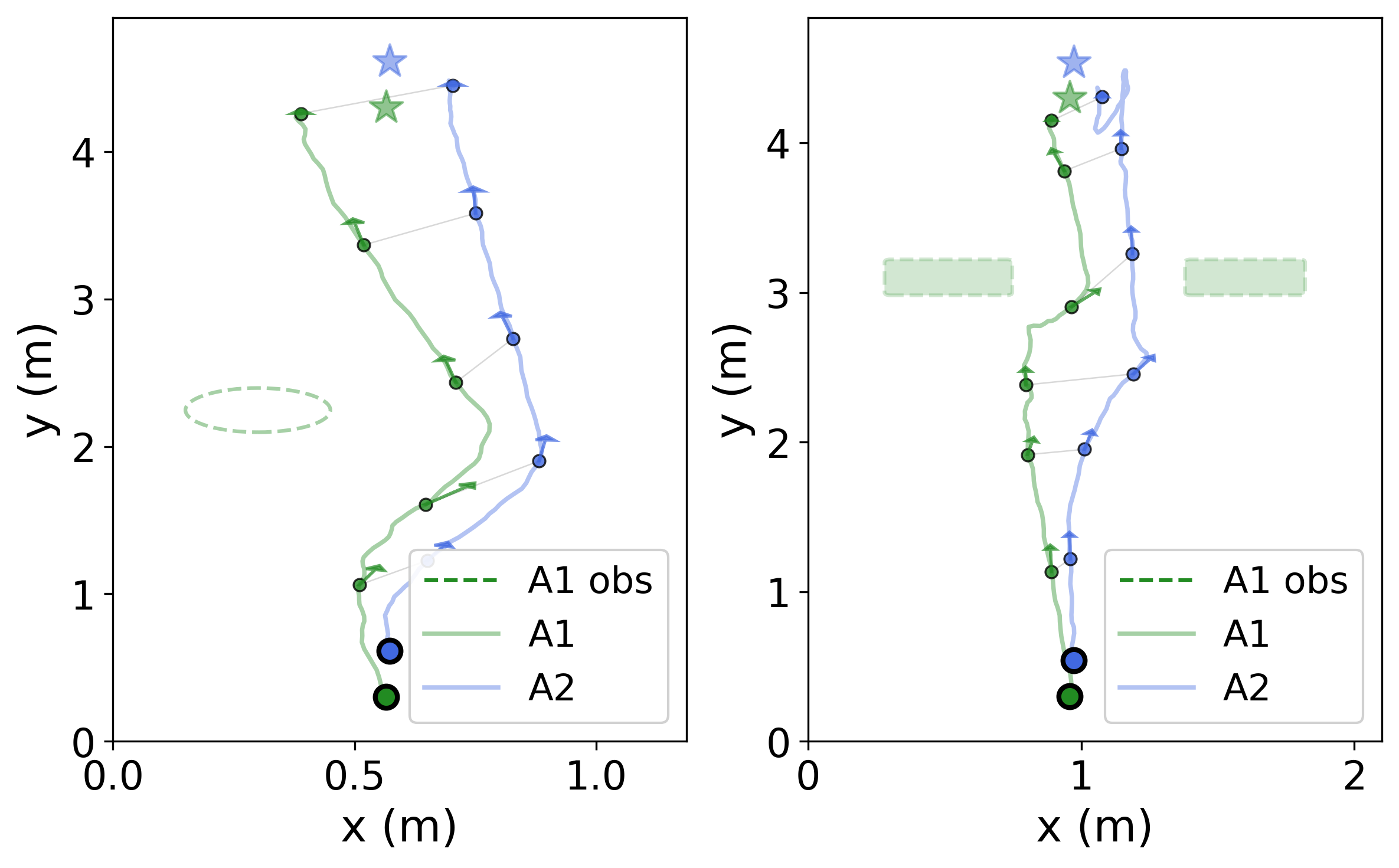}\vspace{-0.4em}
    \caption{Trajectory plots of the hardware experiments depicted in Figure \ref{fig:hardware_keyframes}. The plots confirm that agents were able to avoid obstacles, maintain formation constraints, and reach their goals.}
    \label{fig:hardware_trajectories}\vspace{-0.6em}
\end{figure}

\subsection{Hardware Experiments}

To test the real-world applicability of our method in terms of speed and robustness, we conduct hardware experiments on two Unitree Go2 quadruped robots. Each robot is modeled as a double integrator whose position is tracked by motion capture cameras. The filtered safe action of each robot is directly observed by the other robot. To respect physical hardware limitations, we introduce a maximum velocity constraint in the safety filter, similar to the multi-team example.



In both scenarios shown in Figure \ref{fig:hardware_keyframes}, the rear agent knows about the obstacle and is required to demonstrate it to the front agent. In the left image, a circular object was placed in the path of the agents. In the right image, there are two walls, whose edges are represented by circles, separated by a gap. In both cases, the front agent is able to learn about the obstacles before collision and navigate safely. Obstacle avoidance and formation keeping are confirmed by the trajectory plots in Figure \ref{fig:hardware_trajectories}. These hardware experiments show our method can be run in a real-time, receding-horizon fashion.

\section{Conclusion}

We propose a novel, decentralized multi-agent framework in which agents infer each other’s constraints via their safety-filtered actions. We derive a closed-form solution when only the obstacle constraint is active and a provably convergent Newton solver when there are multiple active constraints. This constraint inference method is combined with a decentralized planning method that provides provable safety guarantees when the agents start in the safe set and the formation distance is sufficiently smaller than the constraint activation distance. We validate our method with Monte Carlo simulations, scale it to multi-agent and multi-team interactions, and demonstrate its real-world applicability and safety with hardware experiments on quadruped robots. 

\bibliographystyle{ieeetr}
\bibliography{references}

\end{document}